\documentclass{llncs}
\usepackage[T1]{fontenc}
\usepackage{graphicx}
\usepackage{xcolor}
\usepackage[pdfencoding=auto]{hyperref}
\hypersetup{
	colorlinks=true,
	linkcolor=blue,
	citecolor=blue,    
	urlcolor=blue
}

\usepackage{amsmath}
\usepackage{amsfonts}
\usepackage{amssymb}

\usepackage{amsthm}

\usepackage{mathtools}
\usepackage{subcaption}

\usepackage{lineno,fnlineno}
\usepackage{float} 
\usepackage{underscore}

\usepackage{tikz}
\usetikzlibrary{shapes,shadows,arrows,graphs,positioning, calc}
\usepackage{adjustbox}
\usepackage{stmaryrd}
\usepackage[noend]{algpseudocode}
\usepackage{algorithm}

\usepackage{graphicx}
\usepackage{xspace}

\usepackage{comment}
\usepackage{todonotes}

\usepackage{tabularx}

\newcolumntype{b}{>{\centering\arraybackslash}X}
\newcolumntype{s}{>{\hsize=.5\hsize}b}

\usepackage{booktabs}
\usepackage{multirow}

\usepackage{thm-restate}

\floatname{algorithm}{Algorithm}

\newcommand{\MCBI}{\textsc{MCBI}\xspace}
\newcommand{\MMCBI}{\textsc{max-MCBI}\xspace}

\newcommand{\MI}{MI\xspace}

\newcommand{\mcb}[1]{\mathcal{MCB}(#1)}

\title{Maximizing Minimum Cycle Bases Intersection}
\author{Ylène Aboulfath \inst{1} \and Dimitri Watel \inst{2,3} \and Marc-Antoine Weisser \inst{4} \and Thierry Mautor \inst{1} \and Dominique Barth \inst{1} }

\authorrunning{Aboulfath Y. et al.} 
\institute{DAVID, Université Versailles Saint-Quentin-En-Yvelines, France \email{ylene.aboulfath@uvsq.fr}, \email{dominique.barth@uvsq.fr}, \email{thierry.mautor@uvsq.fr} \and SAMOVAR, Evry, France \and ENSIIE, Evry, France \email{dimitri.watel@ensiie.fr} \and  LISN, CentraleSupélec, France \email{marc-antoine.weisser@centralesupelec.fr}}

\begin{document}
	\maketitle
	
	\begin{abstract}
		We address a specific case of the matroid intersection problem: given a set of graphs sharing the same set of vertices, select a minimum cycle basis for each graph to maximize the size of their intersection. We provide a comprehensive complexity analysis of this problem, which finds applications in chemoinformatics. We establish a complete partition of subcases based on intrinsic parameters: the number of graphs, the maximum degree of the graphs, and the size of the longest cycle in the minimum cycle bases. Additionally, we present results concerning the approximability and parameterized complexity of the problem.
	\end{abstract}
	
	\section{Introduction}
	\label{sec:intro}
	
	In chemoinformatics and bioinformatics, a molecular dynamics trajectory represents evolution of the 3D positions of atoms constituting a molecule, thus forming a sequence of molecular graphs derived from these positions at discrete time intervals. These graphs share the same vertex set (atoms) but vary in their edges, particularly those representing hydrogen bonds, which may appear or disappear over time, unlike covalent bonds, which are persistent. A research objective is to characterize the evolution of molecular structure during the trajectory~\cite{Bougueroua18,LG16}.
	
	In many studies, the structure of a molecule is intricately linked to the interactions among elementary cycles within its associated graph. Usually, only short cycles intervenne in the characterization of the molecule, therefore, the structure is commonly represented by a minimum cycle basis of the graph \cite{BV13,GS01,Ilemo2019,KZ09,VP97}. Cycle bases are a concise representation of cycles within an undirected graph, and finding a minimum cycle basis ({\em i.e.} minimizing the total weight of cycles in the base) can be done in polynomial time \cite{AR09,HD87}. Thus, given a sequence of molecular graphs modeling the trajectory, to evaluate the conservation of the molecular structure during the trajectory \cite{Aboulfath2024}, we seek to obtain a minimum cycle basis for each graph such that they have the most cycles in common. 
	
	In this context, we define the problem, referred to as \MMCBI, as follows: given a set of $k$ graphs $\{G_1, G_2, \dots, G_k\}$ with the same vertex set, find for each graph a minimum cycle basis such that the size of their overall intersection is maximum. Note that \MMCBI is a special case of the matroid intersection problem (\MI) wherein, given $k$ matroids with the same ground set $C$, we search for one independent set in each matroid such that the size of their intersection is maximum. This is primarily because the set of cycles in an undirected graph forms a vector space. \MI is NP-Complete when $k = 3$ \cite{welsh2010matroid} but polynomial in $|C|$ when $k = 2$ \cite{Edmonds2003} and $\frac{1}{k}$-approximable \cite{korte1978analysis}. Transferring this positive results to \MMCBI can be achieved once we address the challenge posed by the potentially exponential number of cycles in a graph.
	
	\textbf{Our contributions.}
	In this paper, we exploit the distinctive features of our specific instance of \MI to establish a NP-Hard/Polynomial partition of subcases based on intrinsic parameters. Additionally, we investigate the parameterized complexity and approximability of the problem \MMCBI and its decision version \MCBI. In the decision version, given $k$ graphs and a non-negative integer $K$, the objective is to determine whether there exists a minimum cycle basis for each graph such that the size of their intersection is greater than $K$ or not.
	
	We distinguish four intrinsic parameters: the number of graphs $k$; the maximum size $\gamma$ of the cycles in a minimum cycle basis of any graph $G_i$; the maximum degree $\Delta$ in any graph $G_i$; and the decision integer $K$. The first parameter $k$ directly arises from the complexity of \MI which is polynomial for $k=2$ and $NP-complete$ otherwise. Parameters $\gamma$ and $\Delta$ arise from our application. Those are classical parameters studied in molecular contexts. Finally, $K$ is a classical parameter in parameterized complexity.
	
	The results are summarized in Table~\ref{tab:results}. Note that the case $\Delta = 2$ is trivially polynomial as each graph is a set of disjoint cycles. Considering the parameterized complexity and the approximability, few questions remain open.
	
	\begin{table}[ht!]
		\centering
		\caption{Summary of the results of this paper. The hardness results hold true only when the parameters in the \emph{blank cells} are not fixed and remain valid even when the parameters in the \emph{non-blank cells} are fixed to the specified values (or higher). The polynomial results hold true only when the parameters in the \emph{non-blank cells} are fixed to the specified value (or lower) and remain valid regardless of the values of the parameters in the \emph{blank cells}. Additionally, parameters in the parameterized results are indicated with a cross.
		}
		\label{tab:results}
		\setlength{\tabcolsep}{10pt}
		\renewcommand{\arraystretch}{0.8}
		\begin{tabular}{cccc|cc|c}
			\toprule
			$k$ & $\Delta$ & $\gamma$ & $K$ & \MCBI & \MMCBI & \\
			\midrule
			
			\multirow{2}{*}{3} & 4 & 4 & & \multirow{2}{*}{NP-Complete} & \multirow{2}{*}{NP-Hard} & \multirow{7}{*}{Theorem~\ref{theo:hardness}}\\
			& 3 & 5 & & & & \\
			\cmidrule(lr){1-6}
			
			\multirow{2}{*}{} & 4 & 4 & \multirow{2}{*}{-} & \multirow{2}{*}{-} & \multirow{2}{*}{$\frac{1}{k}$ Inapprox.} & \\
			& 3 & 5 & & & & \\
			\cmidrule(lr){1-6}
			
			\multirow{2}{*}{} & 4 & 4 & \multirow{2}{*}{$\times$} & \multirow{2}{*}{W[1]-Hard} & \multirow{2}{*}{-} &\\
			& 3 & 5 & & & & \\
			\midrule
			
			& 2&   & & P & P & Section~\ref{sec:intro}\\
			&  & 3 & & P & P & Theorem~\ref{theo:gamma:3:poly}\\
			& 3& 4 & & P & P & Theorem~\ref{theo:gamma:4:delta:3:poly}\\
			\cmidrule{7-7}
			2 &  &   & & P & P & \multirow{3}{*}{\vspace{-0.3cm}Theorem~\ref{theo:k2:poly:approx:K:XP}}\\
			\cmidrule{1-6}
			&  &  & - & - & $\frac{1}{k}$ Approx. &\\
			\cmidrule{1-6}
			&  &  & $\times$ & XP & & \\
			\bottomrule
		\end{tabular}
	\end{table}

	Section~\ref{sec:mcb} and \ref{sec:matroid} are dedicated to formal definitions and the proof that \MMCBI is indeed a subproblem of \MI. As a result we get the $\frac{1}{k}$- approximation algorithm, the polynomial case when $k = 2$ and the belonging of \MCBI to XP with respect to $K$. In Section~\ref{sec:hardness}, we prove the hardness results. Finally Section~\ref{sec:lowgammadeltacases} is dedicated to the cases where $\gamma = 3$ or where $\gamma = 4$ and $\Delta = 3$.

	\section{Formal Definitions}
	\label{sec:mcb}
	
	We adopt a standard definition of a cycle in a graph $G$ as any subgraph in which each vertex has an even degree\footnote{It's important to note that a cycle can then be composed of elementary cycles, which may seem counter-intuitive at first.}~\cite{HD87}. The sum of two cycles, denoted as $c_1 \oplus c_2$, is the subgraph containing the set of edges present in one and only one of the two cycles. In this paper, when $D$ represents a set of cycles, we denote the sum of cycles as $\bigoplus D = \bigoplus_{d \in D} d$. Note also that if $\bigoplus D = c$ then $\bigoplus D \oplus c = 0$. This general definition of cycles with the sum operation $\oplus$ defines a vector space in the field $\mathbb{Z}/2\mathbb{Z}$. A cycle basis of a graph $G$ is a linearly independent set $B$ of cycles that spans the cycle space of $G$. The terms \emph{span} and \emph{linearly independent} refer to the classical linear algebra definitions.
	\begin{itemize}
		\item $B$ is linearly independent if, for all $B' \subseteq B$, $\bigoplus B' \neq 0$.
		\item $B$ spans a cycle $d$ if there exists $B' \subseteq B$ such that $\bigoplus B' = d$. If $B$ is a basis, then the subset $B'$ is unique for each cycle $d$.
	\end{itemize}
	We define $\lambda_{B} : B \times \mathcal{C} \rightarrow \{0,1\}$ as the function such that if $B'$ is the subset of $B$ with $\bigoplus B' = d$, then $c \in B'$ if and only if $\lambda_{B}(c, d) = 1$. 
	
	The weight of a cycle is defined as its number of edges, denoted for a cycle $c$ by $\omega(c)$. The weight of a cycle basis $B$ is given by $\sum_{c \in B} \omega(c)$. Therefore, a minimum cycle basis is a cycle basis that minimizes its weight. The set of minimum cycle bases of a graph $G$ is denoted by~$\mcb{G}$. Polynomial time algorithms for finding a minimum cycle basis have been proposed~\cite{AR09,HD87}.
	
	\begin{problem}[\MMCBI] 
		Given a set of $k$ graphs $G_1, G_2, \dots, G_k$, find a subset of cycles $B$ of $\bigcap_{i = 1}^k G_i$ such that, for all $i \in \llbracket 1; k \rrbracket$, there exists $B_i \in \mathcal{MCB}(G_i)$ with $B \subseteq B_i$, and maximizing $|B|$.
	\end{problem}
	The decision problem \MCBI associated with \MMCBI is, given an integer $K$, to determine if there exists a solution with $|B| \geq K$.
	
	The rest of this section proves that the search for $B$ can be performed within a polynomial-size subset of cycles. We begin by presenting common lemmas on minimum cycle bases. Their proofs are provided in the Appendix~\ref{app:proofs}.
	
	\begin{restatable}{lemma}{exchangecyclesinbasis}
		\label{lem:exchangecyclesinbasis}
		Given $B$ a cycle basis of a graph $G$ with two cycles $c_1 \in B$ and $c_2 \notin B$, if $\lambda_B(c_1, c_2)~=~1$ then $(B \setminus \{c_1\}) \cup \{c_2\}$ is a cycle basis of $G$.
	\end{restatable}
	
	\begin{restatable}{lemma}{generatecyclebylowerweight}
		\label{lem:generatecyclebylowerweight} 
		$B \in \mcb{G}$ if and only if $B$ is a cycle basis and for $c_1, c_2$ with $c_1 \in B$ and $c_2 \not\in B$ such that $\lambda_B(c_1, c_2) = 1$, we have $\omega(c_1) \leq \omega(c_2)$. 
	\end{restatable}
	
	\begin{restatable}{lemma}{basisexchangeaxiom}
		\label{lem:basisexchangeaxiom} 
		If $B_1, B_2 \in \mcb{G}$, for every $c_1 \in B_1 \backslash B_2$, there exists $c_2 \in B_2 \backslash B_1$ such that $(B_1 \backslash \{c_1\}) \cup \{c_2\}  \in \mcb{G}$. 
	\end{restatable}
	
	Given $G$, let $\mathcal{M}(G) = (C, I)$ be the couple where $C$ is the set of cycles of $G$ and $I$ are the subsets $D$ of $C$ such that there exists $B \in \mcb{G}$ with $D \subseteq B$. 
	
	\begin{lemma}
		\label{lem:mg:matroid} 
		$\mathcal{M}(G)$ is a matroid.
	\end{lemma}
	\begin{proof}
		Lemma~\ref{lem:basisexchangeaxiom} proves the basis exchange axiom. As $\mcb{G}$ is not empty (it possibly contains the empty set of $G$ is a tree), then $\mathcal{M}(G)$ satisfies the basis axioms and is then a matroid. 		
	\end{proof}
	
	Note that \MMCBI can be simply rewritten as the search for a maximum-size set of cycles that are independent in each matroid $\mathcal{M}(G_i)$ using the matroid terminology for \emph{independent}. In order to avoid confusion with the \emph{linear independency}, whenever referring to linear algebra independence, we will explicitly use that terminology. This proves that \MMCBI is indeed a subproblem of the matroid intersection problem (\MI). However, we cannot use algorithms dedicated to \MI to prove any polynomial complexity result as, currently, the ground sets of our matroids have exponential size. We address this in the next section; however, we introduce a polynomial time independency oracle for $\mathcal{M}(G)$. 
	
	\begin{lemma}
		\label{lem:mg:independenceoracle} 
		Given a subset $D$ of cycles of $G$, we can check in polynomial time if $D$ is independent in $\mathcal{M}(G)$.
	\end{lemma}
	\begin{proof}
		This can be achieved by running a modified version of the Horton algorithm \cite{HD87}. Given a graph $G = (V, E)$, the Horton algorithm generates a minimum cycle basis by enumerating a list $L$ of $O(|V||E|)$ cycles, then sorting $L$ from the smallest cycles to the largest and finally using a greedy polynomial-time procedure to build a minimum cycle basis. To adapt this algorithm, we introduce a modification. Before the sorting step, we replace $L$ with $D \cup L$ and during the sorting process, in case of a tie, cycles from $D$ are given priority. The set $D$ is independent if and only if all cycles of $D$ are in the resulting basis.
	\end{proof}

	\section{Case with $k = 2$, Approximability and Parameterized Complexity}
	\label{sec:matroid}
	
	With Theorem~\ref{theo:candidateset}, we prove we may only focus on a polynomial subset of cycles.
	
	\begin{algorithm}
		\caption{Building the list of candidate cycles containing an optimal solution of $\MMCBI$}
		\label{algo:candidatecycle}
		\begin{algorithmic}[1]
			\Function{CandidatesList}{$G_1, G_2, \dots, G_k$}
			\State $L \gets \emptyset$
			\State $G = (V, E) \leftarrow \bigcap_{i = 1}^k G_i$
			\For{$u \in V$, $(v, w) \in E$} \label{algoline:candidatecycle:1} add to $L$ the cycle consisting in the edge $(v, w)$, a shortest path from $u$ to $v$ and a shortest path from $u$ to $w$ in $G$, if such a cycle is elementary.
			\EndFor
			\For{$(u, v) \in E$, $(w, x) \in E$} \label{algoline:candidatecycle:2} add to $L$ two cycles consisting in $(u, v)$, $(w, x)$, one with the shortest paths from $u$ to $w$ and from $v$ to $x$ in $G$, and a second with the shortest paths from $u$ to $x$ and from $v$ to $w$ in $G$, if such a cycle is elementary.
			\EndFor
			\State \Return $L$ 
			\EndFunction
		\end{algorithmic}
	\end{algorithm}
	
	\begin{theorem}
		\label{theo:candidateset}
		Given an instance $\{G_1, G_2, \dots, G_k\}$ of \MMCBI and let $L$ be the list returned by the function \textsc{CandidatesList}$(G_1, G_2, \dots, G_k)$ described in Algorithm~\ref{algo:candidatecycle}, there exists an optimal solution $B^*$ such that $B^* \subseteq L$.
	\end{theorem}
	\begin{proof} 
		Let $B^*$ be any optimal solution and $B_i$ a minimum cycle basis of $G_i$ containing $B^*$ for every $i \in \llbracket 1; k \rrbracket$.
		
		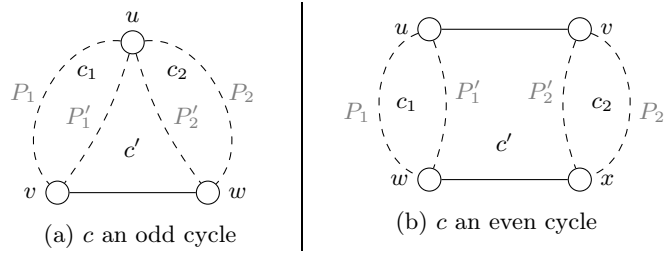
\begin{figure}[ht!]
			\tikzset{tinoeud/.style={draw, circle, minimum height=0.01cm}}
			\tikzstyle{cycle}=[]
			\tikzstyle{chemin}=[dashed]
			\tikzstyle{invisible}=[opacity=0]
			\centering
			\begin{subfigure}[r]{0.3 \linewidth}
				\begin{tikzpicture}
					\node[tinoeud, label={above:$u$}] (V1) at (0,0) {};
					\node[tinoeud, label={left:$v$}] (V2) at (-1,-2) {};
					\node[tinoeud, label={right:$w$}] (V3) at (1,-2) {};
					
					\draw[chemin] (V1) to[bend right=70] node[midway, left] {\color{gray} $P_1$} (V2);
					\draw[chemin] (V1) to[bend left=70] node[midway, right] {\color{gray} $P_2$} (V3);
					
					\draw[chemin] (V1) to[bend left=10] node[midway, left, xshift=-1, yshift=2] {\color{gray} $P_1'$} (V2);
					\draw[chemin] (V1) to[bend right=10] node[midway, right, xshift=1, yshift=2] {\color{gray} $P_2'$} (V3);
					
					\draw (V2) -- (V3);
					
					\draw[cycle] (-0.6, -0.4) node {$c_1$};
					\draw[cycle] (0, -1.4) node {$c'$};
					\draw[cycle] (0.6, -0.4) node {$c_2$};
				\end{tikzpicture}
				\caption{$c$ an odd cycle}
				\label{fig:candidateset:a}
			\end{subfigure}
			~ \vline ~~~
			\begin{subfigure}[l]{0.35 \linewidth}
				\begin{tikzpicture}
					\node[tinoeud, label={left:$u$}] (V1) at (-1,0) {};
					\node[tinoeud, label={right:$v$}] (V2) at (1,0) {};
					\node[tinoeud, label={left:$w$}] (V3) at (-1,-2) {};
					\node[tinoeud, label={right:$x$}] (V4) at (1,-2) {};
					
					\draw[chemin] (V1) to[bend right=70] node[midway, left, yshift=-3] {\color{gray} $P_1$}(V3);
					\draw[chemin] (V2) to[bend left=70] node[midway, right, yshift=-3] {\color{gray} $P_2$} (V4);
					
					\draw[chemin] (V1) to[bend left=20] node[midway, right, yshift=5] {\color{gray} $P_1'$} (V3);
					\draw[chemin] (V2) to[bend right=20]  node[midway, left, yshift=5] {\color{gray} $P_2'$} (V4);
					
					\draw (V1) -- (V2);
					\draw (V3) -- (V4);
					
					\draw[cycle] (-1.3, -1) node {$c_1$};
					\draw[cycle] (0, -1.5) node {$c'$};
					\draw[cycle] (1.3, -1) node {$c_2$};
				\end{tikzpicture}
				\caption{$c$ an even cycle}
				\label{fig:candidateset:b}
			\end{subfigure}
			
			\caption{Example of cycles such that $c_1 \oplus c' \oplus c_2 = c$. }
			
			\label{fig:candidateset}
		\end{figure}
		
		Let us suppose that there exists a cycle $c \in B^* \backslash L$. Hereinafter, we prove that $c$ can be replaced by a cycle $c' \in L \backslash B^*$ such that $B^* \backslash \{c\} \cup \{c'\}$ is still optimal.
		
		Assuming $c$ is an odd cycle containing an edge $(v, w)$. Let $P_1$ and $P_2$ be the two paths included in $c$ connecting respectively $v$ and $w$ to the same node $u$ such that $|P_1| = |P_2|$. At line~\ref{algoline:candidatecycle:1} of Algorithm~\ref{algo:candidatecycle}, when the loop enumerates $u$ and $(v, w)$, we obtain a cycle $c'$ that is added to $L$ if it is elementary. As depicted by Figure~\ref{fig:candidateset:a}, there exist two (possibly non-elementary) cycles $c_1$ and $c_2$ such that $c = c_1 \oplus c' \oplus c_2$. As $P_1'$ and $P_2'$ are shortest paths and as $|P_1| = |P_2| = (|c| - 1) / 2$, we have $\omega(c_1) < \omega(c), \omega(c') \leq \omega(c)$ and $\omega(c_2) < \omega(c)$. If $c'$ is not elementary, then $c$ is the sum of strictly smaller cycles, this contradicts Lemma~\ref{lem:generatecyclebylowerweight}. Consequently, $c'$ is elementary and is added to $L$. For the same reason, $\omega(c') = \omega(c)$.
		
		For all $i \in \llbracket 1 ; k \rrbracket$, let $B_i$ be a minimum cycle basis of $G_i$ containing $B^*$. By Lemma~\ref{lem:generatecyclebylowerweight}, every $d$ such that $\lambda_{B_i}(d, c_1) = 1$ has a weight $\omega(d) \leq \omega(c_1)$. Similarly for $c_2$. Thus those cycles cannot be $c'$ and $c$. As $c' = \bigoplus \{d : \lambda_{B_i}(d, c_1) \} \oplus \bigoplus \{d : \lambda_{B_i}(d, c_2) \} \oplus c$, then, $c' \not\in B_i$, otherwise $B_i$ is not linearly independent. 
		
		Consequently, $\omega(c) = \omega(c')$ and $\lambda_{B_i}(c, c') = 1$ for all basis $B_i$. We can then replace $c$ by $c'$ in $B_i$ by Lemma~\ref{lem:exchangecyclesinbasis}. The same property occurs for even cycles (see Figure~\ref{fig:candidateset:b}). This operation can be repeated until $B^* \subseteq L$. As the size of $B^*$ is unchanged, we get another optimal solution.		
	\end{proof}
	
	\MMCBI can now be seen as a subproblem of the matroid intersection problem, \MI. We deduce the following theorem.

	\begin{theorem}
		\label{theo:k2:poly:approx:K:XP}
		\MCBI and \MMCBI are polynomial when $k = 2$, \MMCBI is $\frac{1}{k}$-approximable and, finally \MCBI is XP with respect to $K$.
	\end{theorem}
	\begin{proof}
		\MMCBI consists in solving \MI in $\{\mathcal{M}(G_i), i \in \llbracket 1; k\rrbracket\}$. By Theorem~\ref{theo:candidateset}, we can restrict each matroid to $L$, the cycles output by Algorithm~\ref{algo:candidatecycle}. Let $\mathcal{M}(G_i)_{|L}$ be the resulting restricted matroid. By Lemma~\ref{lem:mg:independenceoracle} and Theorem~\ref{theo:candidateset} $\mathcal{M}(G_i)_{|L}$ is a matroid with a polynomial size ground set and a polynomial time independence oracle.  In that case \MI is polynomial when $k = 2$ \cite{Edmonds2003} and is $\frac{1}{k}$-approximable \cite{korte1978analysis}. Finally if $K$ is fixed, we can simply enumerate every subset of $L$ of size $K$ and check if that subset is independent in all the matroids. 
	\end{proof}

	\section{Hardness of \MCBI}
	\label{sec:hardness}
	This section gives the hardness proofs for \MMCBI and \MCBI. The latter is in NP as by Lemma~\ref{lem:mg:independenceoracle} we can check independence in polynomial time. 
	
	\begin{theorem}
		\label{theo:hardness}
		\MCBI is NP-Complete even if $k = 3$. In addition, \MCBI is W[1]-Hard with respect to $K$. Moreover, unless P = NP, for every $\varepsilon > 0$, there is no polynomial approximation algorithm with ratio $\frac{1}{k^{1 - \varepsilon}}$ for $\MMCBI$. All those results remain true even if $\Delta = 3$ and $\gamma = 5$ or if $\Delta = 4$ and $\gamma = 4$. 
	\end{theorem}
	\begin{proof}
		We provide a reduction from the \emph{Maximum Independent Set} in a graph that consists, given a graph $H = (V, E)$ and an integer $K'$, in the search for an independent set of $H$ of size at least $K'$, that is a subset in which no pair of nodes is linked by an edge in $E$. In order to avoid confusions with the word \emph{independent}, we then use the \emph{Stable set} terminology instead of Independent set. 
		
		The reduction does not depend on the value of $\gamma$ except for a useful procedure described hereinafter and in Figure~\ref{fig:np-hardness:CONN}. Let $c_1 = (u_1, u_2, \dots, u_l)$ and $c_2 = (v_1, v_2, \dots, v_l)$ be two disjoint cycles where $l \in \{4, 5\}$. The procedure $CONN(c_1, c_2)$ connects the nodes of the two cycles. We add the edges $(u_i, v_i)$ for $i \in \llbracket 1; l \rrbracket$ and, if $l = 4$, we add the edges $(u_1, v_2), (u_2, v_3), (u_3, v_4)$ and $(u_4, v_1)$. Note that, by Lemma~\ref{lem:generatecyclebylowerweight}, after using the procedure, no more minimum cycle basis containing $c_1$ and $c_2$ at the same time. 
		
		\begin{figure}[ht!]
			\centering
			\begin{subfigure}{0.45\linewidth}
				\centering
				\begin{tikzpicture}[scale=0.5]
					\tikzset{tinoeud/.style={draw, circle, minimum height=0.01cm}}
					\tikzset{titerminal/.style={draw, rectangle, minimum height=0.01cm}}
					
					\node[tinoeud] (V1) at (0,0) {};
					\node[tinoeud] (V2) at (2,0) {};
					\node[tinoeud] (V3) at (2,2) {};
					\node[tinoeud] (V4) at (0,2) {};
					
					\node[tinoeud] (W1) at (-1,-1) {};
					\node[tinoeud] (W2) at (3,-1) {};
					\node[tinoeud] (W3) at (3,3) {};
					\node[tinoeud] (W4) at (-1,3) {};
					
					\draw (V1) -- (V2) -- (V3) -- (V4) -- (V1);
					\draw (W1) -- (W2) -- (W3) -- (W4) -- (W1);
					
					\draw (V1) -- (W1) -- (V2);
					\draw (V2) -- (W2) -- (V3);
					\draw (V3) -- (W3) -- (V4);
					\draw (V4) -- (W4) -- (V1);
					
					\draw (1, 1) node {$c_2$};
					\draw (-2, 1) node {$c_1$};
				\end{tikzpicture}
				\caption{$\gamma = 4$}
				\label{subfig:np-hardness:CONN:1}
			\end{subfigure}
			\begin{subfigure}{0.45\linewidth}
				\centering
				\begin{tikzpicture}[scale=0.55]
					\tikzset{tinoeud/.style={draw, circle, minimum height=0.01cm}}
					\tikzset{titerminal/.style={draw, rectangle, minimum height=0.01cm}}
					
					\node[tinoeud] (V1) at (0:1) {};
					\node[tinoeud] (V2) at (72:1) {};
					\node[tinoeud] (V3) at (144:1) {};
					\node[tinoeud] (V4) at (216:1) {};
					\node[tinoeud] (V5) at (288:1) {};
					
					\node[tinoeud] (W1) at (0:2) {};
					\node[tinoeud] (W2) at (72:2) {};
					\node[tinoeud] (W3) at (144:2) {};
					\node[tinoeud] (W4) at (216:2) {};
					\node[tinoeud] (W5) at (288:2) {};
					
					\draw (V1) -- (V2) -- (V3) -- (V4) -- (V5) -- (V1);
					\draw (W1) -- (W2) -- (W3) -- (W4) -- (W5) -- (W1);
					
					\draw (V1) -- (W1);
					\draw (V2) -- (W2);
					\draw (V3) -- (W3);
					\draw (V4) -- (W4);
					\draw (V5) -- (W5);
					
					\draw (0, 0) node {$c_2$};
					\draw (-2.5, 0) node {$c_1$};
				\end{tikzpicture}
				\caption{$\gamma = 5$}
				\label{subfig:np-hardness:CONN:2}
			\end{subfigure}
			\caption{Illustration of the procedure $CONN(c_1, c_2)$}
			\label{fig:np-hardness:CONN}
		\end{figure}
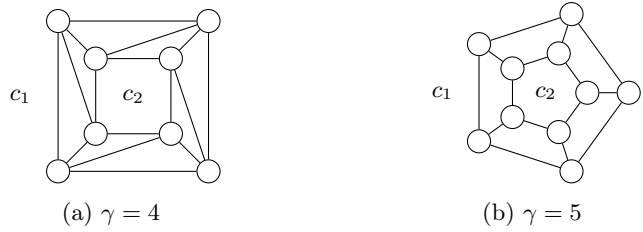

		We now describe a general reduction with more than 3 graphs and then show how to reduce $k$. Let then $(H = (V, E), K')$ be an instance of the maximum Stable set problem. We create an instance of \MCBI as follows. We set $l$ either to 4 or 5. We set $K = K'$. For each edge $e \in E$, we build a graph $G_e$. For each node $v \in V$, we add a cycle $c_v$ of size $l$ to all the graphs. All the cycles are disjoint. Finally, if $e = (u, v)$, then, in the graph $G_e$, we connect $c_u$ and $c_v$ using the procedure $CONN(c_u, c_v)$. Every other cycle of $G_e$ is not connected to the rest of the graph and is its own connected component. 
		
		Note that $l = \gamma$ and if $\gamma = 4$, then $\Delta = 4$ and if $\gamma = 5$ then $\Delta = 3$. Note also that only the cycles $\{c_v | v \in V\}$ belong to the intersection of the graphs $\{G_e | e \in E\}$. A feasible solution contains only those cycles.
		
		Given an edge $e = (u, w) \in E$, the graph $G_e$ contains exactly two minimum cycles bases. If $\gamma = 4$ (resp. $5$), let $D$ be the set of triangles (resp. squares) connecting $c_u$ and $c_w$. Then $\mcb{G_e} = \{D \cup \{c_v | v \in V \backslash \{w\}\}, D \cup \{c_v | v \in V \backslash \{u\}\}\}$. As a consequence, for every subset of nodes $V' \subset V$, the set $\{c_v | v \in V'\}$ is independent in $G_e$ if and only if $u \not\in V'$ or $w \not\in V'$. Thus, there exists a feasible set of cycles of size $K$ if and only if there exists a stable set of size $K$ in $H$. As the Maximum Stable Set is NP-Complete and W[1]-Hard with respect to $K'$, this reduction proves the NP-Completeness of \MCBI and the W[1]-Hardness with respect to $K$. 
		
		Now note that, if $e = (u, v)$ and $f = (u', v')$ are two non-incident edges, then instead of $G_e$ and $G_f$ we can add to the instance a graph $G_{(e, f)}$ containing the union of $G_e$ and $G_f$. As the two edges are not incident, we have that a feasible set of cycles cannot contain $c_u$ and $c_v$ at the same time and the same for $c_{u'}$ and $c_{v'}$. Note that this transformation does not change the values of $\gamma$ and $\Delta$. We can extend this idea to any matching $M$ of $H$. By the Vizing Theorem \cite{MG92}, using a polynomial greedy algorithm, edges of $H$ can be covered with $|V|$ disjoint matchings. This reduces the number $k$ of graphs from $|E|$ to $|V|$. By \cite{ZD06}, unless P = NP, for all $\varepsilon > 0$, there is no polynomial approximation with ratio $\frac{1}{|V|^{1 - \varepsilon}}$ for the Maximum Stable Set problem. This proves the inapproximability result.
		
		The Maximum Stable Set remains NP-Complete even if $H$ has degree at most 3 and if any path linking two nodes with degree 3 contains at least 3 edges \cite{Murphy1992}. Such a graph is always 3-edge-colorable. This proves the NP-Completeness of \MCBI when $k = 3$.	
	\end{proof}

	\section{Cases $\gamma = 3$, and $\gamma = 4, \Delta = 3$}
	\label{sec:lowgammadeltacases}
	
	We describe two useful lemmas to prove that we can focus on the cycles size by size independently. In this section, we call $L$ the list returned by Algorithm~\ref{algo:candidatecycle}, and we denote by $T(G)$ and $S(G)$ the triangles and squares of a graph $G$.
	
	\begin{lemma}
		\label{lem:exchangesamesizecycle} 
		If $B_1, B_2 \in \mcb{G}$, for every $l \in \mathbb{N}$, $\{c \in B_1 | \omega(c) \neq l\} \cup \{c \in B_2 | \omega(c) = l\} \in \mcb{G}$. 
	\end{lemma}
	\begin{proof}
		Using Lemma~\ref{lem:basisexchangeaxiom}, for every $c \in B_1 \backslash B_2$, there is some $d \in B_2 \backslash B_1$ such that $B_1 \backslash \{c\} \cup \{d\} \in \mcb{G}$. Note that $\omega(c) = \omega(d)$. We can then swap the cycles until all the cycles of $B$ with size $l$ are in $B_2$. After the exchanges, $|\{c \in B | \omega(c) = l \}| = |\{c \in B_2 | \omega(c) = l\}|$, otherwise, the basis exchange axiom would be false for $B$ and $B_2$. Then $B = \{c \in B_1 | \omega(c) \neq l\} \cup \{c \in B_2 | \omega(c) = l\}$. 
	\end{proof}
	
	Using Lemma~\ref{lem:exchangesamesizecycle} we see that all the subset of cycles of same size in two bases are interchangeable. This means that, from two feasible solutions $B_3$ and $B_4$ respectively maximizing the number of triangles and squares, we can build a feasible solution maximizing both of them. We define with $\mathcal{M}(G, l)$ the matroid $\mathcal{M}(G)$ retricted to the cycles of size $l$ in $L$. Maximizing the triangles (respectively squares) consists in finding a maximum set of cycles that are independent in $\mathcal{M}(G, 3)$ (resp. $\mathcal{M}(G, 4)$). The following lemma gives a characterization of the independent sets. Due to lack of space, the proof may be found in Appendix~\ref{app:proofs}.
	
	\begin{restatable}{lemma}{circuit}
		\label{lem:circuit}
		$B$ is independent in $\mathcal{M}(G, l)$ if and only if, for all $D \subseteq B$,\\ $\bigoplus D\not\in span(\text{cycle } c | \omega(c) \leq l - 1)$. 
	\end{restatable}
	
	\begin{theorem}
		\label{theo:gamma:3:poly}
		Given an instance of \MMCBI, one can find a feasible solution maximizing the number of triangles in polynomial time. As a consequence, \MCBI and \MMCBI are polynomial when $\gamma = 3$.
	\end{theorem}
	\begin{proof}
		Because we work with simple graphs, there are no cycles of size 2. By Lemma~\ref{lem:circuit}, any set of linearly independent triangles of $L$ is independent in $\mathcal{M}(G_i, 3)$ for all $i$. We can then start with an empty solution $B$ and add each cycle $c \in T(G) \cap L$ to $B$ if $c \not\in span(B)$.
	\end{proof}
	
	Now we consider the case for $\gamma = 4$ and $\Delta = 3$. An instance where $\gamma = 4$ has triangles and squares in $L$. Contrary to the previous case, we can have two squares that are linearly independent but not together independent in $\mathcal{M}(G_i, 4)$ for all $i$. For instance one of the squares may contain a diagonal in $G_i$ or the sum of the squares may belong to the span of the triangles like in Figure~\ref{subfig:np-hardness:CONN:1}. 
	
	Interestingly, the algorithm we use is almost the same as for $\gamma = 3$. It is described in the proof of Theorem~\ref{theo:gamma:4:delta:3:poly}.
	 
	We consider that $L$ do not contain any square that is spanned by triangles of $T(G)$ as such a square cannot belong to any independent set of $\mathcal{M}(G_i, 4)$. Such cycles can be removed from $L$ in polynomial time.
	
	In the following lemmas, we use the matroid terminology of \emph{circuit}. By Lemma~\ref{lem:circuit}, a circuit of $\mathcal{M}(G, 4)$ is a subset $C \subseteq S(G)$ such that $\bigoplus C \in span(T(G))$ but, for all subsets $C'$ of $C$, we have $\bigoplus C' \not\in span(T(G))$. We now caracterize the circuits of the graphs when $\Delta = 3$. 
	
	\begin{lemma}
		\label{lem:gamma:4:delta:3:1:intersects}
		Let $G \in \{G_i, i \leq k\}$ with $\Delta = 3$ and $s_1 = (a_1, b_1, c_1, d_1) \in L$, $s_2 = (a_2, b_2, c_2, d_2) \in L$, with $s_1 \neq s_2$ and $t_3 = (a_3, b_3, c_3) \in T(G)$. Figure~\ref{fig:gamma:4:delta:3:1} gives the possible intersections of $s_1$ and $s_2$, and of $s_1$ and $t_3$, up to an isomorphism.
	\end{lemma}
	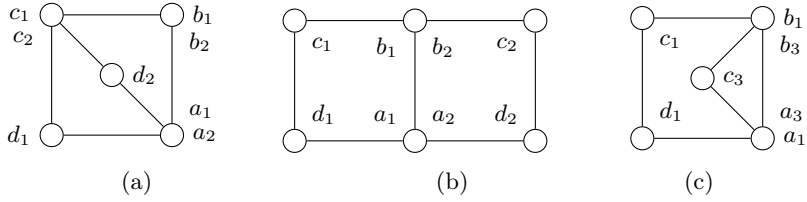
\begin{figure}[t]
		\centering
		\begin{subfigure}{0.3\linewidth}
			\begin{tikzpicture}[scale=0.8]
				\tikzset{tinoeud/.style={draw, circle, minimum height=0.01cm}}
				\tikzset{titerminal/.style={draw, rectangle, minimum height=0.01cm}}
				
				\node[tinoeud, label={left:$d_1$}] (V1) at (0,0) {};
				\node[tinoeud, label={left:$c_1$}, label={below left:$c_2$}] (V2) at (0,2) {};
				\node[tinoeud, label={right:$b_1$}, label={below right:$b_2$}] (V3) at (2,2) {};
				\node[tinoeud, label={above right:$a_1$}, label={right:$a_2$}] (V4) at (2,0) {};
				\node[tinoeud, label={right:$d_2$}] (V5) at (1,1) {};
				
				\draw (V1) -- (V2) -- (V3) -- (V4) -- (V1);
				\draw (V2) -- (V5) -- (V4);
			\end{tikzpicture}
			\caption{}
			\label{subfig:gamma:4:delta:3:1:1}
		\end{subfigure}
		\begin{subfigure}{0.37\linewidth}
			\begin{tikzpicture}[scale=0.8]
				\tikzset{tinoeud/.style={draw, circle, minimum height=0.01cm}}
				\tikzset{titerminal/.style={draw, rectangle, minimum height=0.01cm}}
				
				\node[tinoeud, label={above right:$a_2$}, label={above left:$a_1$}] (V1) at (0,0) {};
				\node[tinoeud, label={above left:$d_2$}] (V2) at (2,0) {};
				\node[tinoeud, label={below left:$c_2$}] (V3) at (2,2) {};
				\node[tinoeud, label={below right:$b_2$}, label={below left:$b_1$}] (V4) at (0,2) {};
				\node[tinoeud, label={above right:$d_1$}] (V5) at (-2,0) {};
				\node[tinoeud, label={below right:$c_1$}] (V6) at (-2,2) {};
				
				\draw (V1) -- (V2) -- (V3) -- (V4) -- (V1);
				\draw (V1) -- (V5) -- (V6) -- (V4);
			\end{tikzpicture}
			\caption{}
			\label{subfig:gamma:4:delta:3:1:2}
		\end{subfigure}
		\begin{subfigure}{0.15\linewidth}
			\begin{tikzpicture}[scale=0.8]
				\tikzset{tinoeud/.style={draw, circle, minimum height=0.01cm}}
				\tikzset{titerminal/.style={draw, rectangle, minimum height=0.01cm}}
				
				\node[tinoeud, label={above right:$d_1$}] (V1) at (0,0) {};
				\node[tinoeud, label={below right:$c_1$}] (V2) at (0,2) {};
				\node[tinoeud, label={right:$b_1$}, label={below right:$b_3$}] (V3) at (2,2) {};
				\node[tinoeud, label={right:$a_1$}, label={above right:$a_3$}] (V4) at (2,0) {};
				\node[tinoeud, label={right:$c_3$}] (V5) at (1,1) {};
				
				\draw (V1) -- (V2) -- (V3) -- (V4) -- (V1);
				\draw (V4) -- (V5) -- (V3);
			\end{tikzpicture}
			\caption{}
			\label{subfig:gamma:4:delta:3:1:3}
		\end{subfigure}
		\caption{Possible cases of intersection of squares and of triangles.}
		\label{fig:gamma:4:delta:3:1}
	\end{figure}
	\begin{proof}
		Recall that we removed from $L$ the squares generated by triangles of $T(G)$: they do not contain a chord. For the squares intersection, if $s_1$ and $s_2$ have only one common node, that node has degree 4 and 4 common nodes imply two diagonals. Thus they have three common nodes (Figure~\ref{subfig:gamma:4:delta:3:1:1}) or two common nodes without diagonal (Figure~\ref{subfig:gamma:4:delta:3:1:2}). Similarly, if $s_1$ and $t_3$ have 1 or 3 common nodes, there is a contradiction. Without diagonal, this leads to Figure~\ref{subfig:gamma:4:delta:3:1:3}.
	\end{proof}
	
	\begin{lemma}
		\label{lem:gamma:4:delta:3:2:connected}
		Let $G \in \{G_i, i \leq k\}$ with $\Delta = 3$ and $C \subseteq L$ be a circuit of $\mathcal{M}(G_i, 4)$. Then $\bigcup C$ is a connected subgraph of $G$.
	\end{lemma}
	\begin{proof}
		As $C$ is a circuit of $\mathcal{M}(G_i, 4)$, by Lemma~\ref{lem:circuit}, there exists $T \subseteq T(G)$ such that $\bigoplus T = \bigoplus C$. We assume that $T$ is minimal, meaning that for $T' \subsetneq T$, $\bigoplus T \neq 0$. We now show that $\bigcup C \cup \bigcup T$ is a connected graph. Indeed, otherwise, let assume there exists a connected component $G'$ of $\bigcup C \cup \bigcup T$, and let $C'$ and $T'$ be respectively the proper subets of $C$ and $T$ in $G'$. No cycle of $C \backslash C'$ and $T \backslash T'$ intersects the edges of $\bigcup C' \cup \bigcup T'$. As $\bigoplus T = \bigoplus C$ then $\bigoplus T' = \bigoplus C'$. Consequently either $C'$ is empty in which case $\bigoplus T' = 0$ or $C'$ is a dependent proper subset of $C$. The first case contradicts the minimality of $T$ and the second one contradicts the fact that $C$ is a circuit. The only left possibility is that $T' = C' = \emptyset$ meaning that $\bigcup C \cup \bigcup T$ is connected.
		
		Now we assume that $\bigcup C$ is disconnected. In $\bigcup C \cup \bigcup T$, any two connected components of $\bigcup C$ are connected by a chain $P$ included in $\bigcup T$. Let $s_1 = (a_1, b_1, c_1, d_1)$ and $s_2 = (a_2, b_2, c_2, d_2)$ be two squares of $C$ respectively in the first and second component linked by $P$. Let $w_1, w_2, \dots, w_q$ be the nodes of $P$, with $w_1 \in s_1$ and $w_q \in s_2$. Without loss of generality, we state that $w_1 = a_1$ and $w_q = a_2$. Note that we cannot have $q = 1$ as the squares do not intersect.
		
		Each edge $(w_i, w_{i + 1})$ belongs to a triangle $t_i$. As $w_1 = a_1$ has 3 neighbors, $b_1, d_1$ and $w_2$, then $t_1$ is either $(w_1, w_2, b_1)$ or $(w_1, w_2, d_1)$. We assume, wlog, that it is $(w_1, w_2, b_1)$. As a result, $q \neq 2$, otherwise, $w_2 = a_2$ has four neighbors ($a_1, b_1, b_2$ and $d_2$). Thus $q \geq 3$. As $w_2$ has already 3 neighbors, $w_3, a_1$ and $b_1$, the nodes of $t_2$ are $w_2, w_3$ and either $a_1$ or $b_1$. This means that $w_3$ is connected to that node, and then the degree of $a_1$ or $b_1$ is 4. We then have a contradiction.		
	\end{proof}
	
	\begin{lemma}
		\label{lem:gamma:4:delta:3:3:circuits}
		Let $G \in \{G_i, i \leq k\}$ with $\Delta = 3$. Then Figure~\ref{fig:gamma:4:delta:3:2} gives the possible circuits $C$ of $G$ such that $\bigoplus C \neq 0$.
	\end{lemma}
	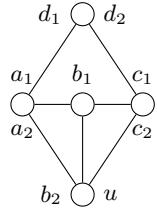
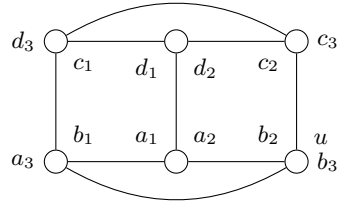
\begin{figure}[t]
		\centering
		\begin{subfigure}{0.45\linewidth}
			\centering
			\begin{tikzpicture}[scale=0.8]
				\tikzset{tinoeud/.style={draw, circle, minimum height=0.01cm}}
				\tikzset{titerminal/.style={draw, rectangle, minimum height=0.01cm}}
				
				\node[tinoeud, label={above:$a_1$}, label={below:$a_2$}] (V1) at (0,0) {};
				\node[tinoeud, label={left:$d_1$}, label={right:$d_2$}] (V2) at (1,1.5) {};
				\node[tinoeud, label={above:$c_1$}, label={below:$c_2$}] (V3) at (2,0) {};
				\node[tinoeud, label={above:$b_1$}] (V4) at (1,0) {};
				\node[tinoeud, label={left:$b_2$}, label={right:$u$}] (V5) at (1,-1.5) {};
				
				\draw (V1) -- (V2) -- (V3) -- (V4) -- (V1);
				\draw (V3) -- (V5) -- (V1);
				\draw (V4) -- (V5);
			\end{tikzpicture}
			\caption{Circuit with two squares.}
			\label{subfig:gamma:4:delta:3:2:1}
		\end{subfigure}
		\begin{subfigure}{0.45\linewidth}
			\centering
			\begin{tikzpicture}[scale=0.8]
				\tikzset{tinoeud/.style={draw, circle, minimum height=0.01cm}}
				\tikzset{titerminal/.style={draw, rectangle, minimum height=0.01cm}}
				
				\node[tinoeud, label={above left:$a_1$}, label={above right:$a_2$}] (V1) at (0,0) {};
				\node[tinoeud, label={right:$b_3$}, label={above left:$b_2$}, label={above right:$u$}] (V2) at (2,0) {};
				\node[tinoeud, label={right:$c_3$}, label={below left:$c_2$}] (V3) at (2,2) {};
				\node[tinoeud, label={below left:$d_1$}, label={below right:$d_2$}] (V4) at (0,2) {};
				\node[tinoeud, label={left:$a_3$}, label={above right:$b_1$}] (V5) at (-2,0) {};
				\node[tinoeud, label={left:$d_3$}, label={below right:$c_1$}] (V6) at (-2,2) {};
				
				\draw (V1) -- (V2) -- (V3) -- (V4) -- (V1);
				\draw (V1) -- (V5) -- (V6) -- (V4);
				\draw (V2) to[bend left] (V5);
				\draw (V3) to[bend right] (V6);
			\end{tikzpicture}
			\caption{Circuit with three squares.}
			\label{subfig:gamma:4:delta:3:2:2}
		\end{subfigure}
		\caption{Possible circuits containing squares in a graph $G$ with degree at most 3. Each square contains four node $(a_i, b_i, c_i, d_i)$.}
		\label{fig:gamma:4:delta:3:2}
	\end{figure}
	\begin{proof}
		As $\bigoplus C \neq 0$, then, there exists $T \subseteq T(G)$ with $T \neq \emptyset$ and $\bigoplus C = \bigoplus T$. There exists at least one triangle $t$ that intersects a square $s_1$ of $C$. By Lemma~\ref{lem:gamma:4:delta:3:1:intersects}, that intersection is the graph depicted by Figure~\ref{subfig:gamma:4:delta:3:1:3}. We set $(a_1, b_1, c_1, d_1)$ and $t = (a_1, b_1, u)$. As we removed from $L$ the squares generated by triangles, $|C| > 1$. By Lemma~\ref{lem:gamma:4:delta:3:2:connected}, there exists another square $s_2$ that intersects $s_1$. Let $s_2 = (a_2, b_2, c_2, d_2)$. In the following, we rename the nodes of $s_2$ so that if $a_1$ (resp. $b_1, c_1, d_1$) is in $s_1 \cap s_2$, then $a_1 = a_2$ (resp. $b_1 = b_2$, $c_1 = c_2$, $d_1 = d_2$). 
		\begin{itemize}
			\item If $s_1$ and $s_2$ have three common nodes as in Figure~\ref{subfig:gamma:4:delta:3:1:1}, then:
			\begin{itemize}
				\item $s_1 \cap s_2 = \{a_1 = a_2, b_1 = b_2, c_1 = c_2\}$ and $d_1 \neq d_2$ (see Figure~\ref{subfig:gamma:4:delta:3:3:1}). Then $u$, $b_1$, $d_1$ and $d_2$ are neighbors of $a_1$, then two of those nodes are equal. Otherwise the degree of $a_1$ is 4. By hypothesis, $d_1 \neq d_2$. We cannot have $d_1 = b_1$ as $s_1$ would not be a square, similarly $b_1 \neq u$. And $d_1 \neq u$ otherwise $s_1$ would contain a diagonal. The only possibilities are $d_2 = u$ or $d_2 = b_1$. In the first case $s_2$ contains a diagonal $(b_2, d_2)$. In the second case, $d_2 = b_2$. Thus there is a contradiction.
				
				\begin{figure}[t]
					\centering
					\begin{subfigure}[c]{0.32\textwidth}
						\centering
						\begin{tikzpicture}[scale=0.8]
							\tikzset{tinoeud/.style={draw, circle, minimum height=0.01cm}}
							\tikzset{titerminal/.style={draw, rectangle, minimum height=0.01cm}}
							
							\node[tinoeud, label={left:$d_1$}] (V1) at (0,0) {};
							\node[tinoeud, label={left:$c_1 c_2$}, 
							] (V2) at (0,2) {};
							\node[tinoeud, label={right:$b_1 b_2$}, 
							] (V3) at (2,2) {};
							\node[tinoeud, label={right:$a_1 a_2$}, 
							] (V4) at (2,0) {};
							\node[tinoeud, label={left:$u$}] (V5) at (3,1) {};
							\node[tinoeud, dotted, label={left:$d_2$}] (V6) at (1,1) {};
							
							\draw (V1) -- (V2) -- (V3) -- (V4) -- (V1);
							\draw (V4) -- (V5) -- (V3);
							\draw[dashed] (V4) -- (V6) -- (V2);
						\end{tikzpicture}
						\caption{}
						\label{subfig:gamma:4:delta:3:3:1}
					\end{subfigure}
					\begin{subfigure}[c]{0.32\textwidth}
						\centering
						\begin{tikzpicture}[scale=0.8]
							\tikzset{tinoeud/.style={draw, circle, minimum height=0.01cm}}
							\tikzset{titerminal/.style={draw, rectangle, minimum height=0.01cm}}
							
							\node[tinoeud, label={left:$d_1 d_2$}, 
							] (V1) at (0,0) {};
							\node[tinoeud, label={left:$c_1$}] (V2) at (0,2) {};
							\node[tinoeud, label={right:$b_1$}] (V3) at (2,2) {};
							\node[tinoeud,label={right:$a_1 a_2$}, 
							] (V4) at (2,0) {};
							\node[tinoeud,label={left:$u$}] (V5) at (3,1) {};
							\node[tinoeud, dotted, label={above:$c_2$}] (V6) at (0.5,1) {};
							\node[tinoeud, dotted, label={above:$b_2$}] (V7) at (1.5,1) {};
							
							\draw (V1) -- (V2) -- (V3) -- (V4) -- (V1);
							\draw (V4) -- (V5) -- (V3);
							\draw[dashed] (V4) -- (V7) -- (V6) -- (V1);
						\end{tikzpicture}
						\caption{}
						\label{subfig:gamma:4:delta:3:3:2}
					\end{subfigure}
					\begin{subfigure}[c]{0.32\textwidth}
						\centering
						\begin{tikzpicture}[scale=0.8]
							\tikzset{tinoeud/.style={draw, circle, minimum height=0.01cm}}
							\tikzset{titerminal/.style={draw, rectangle, minimum height=0.01cm}}
							
							\node[tinoeud, label={left:$d_1 d_2$}] (V1) at (0,0) {};
							\node[tinoeud, label={left:$c_1 c_2$}] (V2) at (0,2) {};
							\node[tinoeud, label={right:$b_1$}] (V3) at (2,2) {};
							\node[tinoeud, label={right:$a_1$}] (V4) at (2,0) {};
							\node[tinoeud, label={left:$u$}] (V5) at (3,1) {};
							\node[tinoeud, dotted, label={right:$a_2$}] (V6) at (1,0.5) {};
							\node[tinoeud, dotted, label={right:$b_2$}] (V7) at (1,1.5) {};
							
							\draw (V1) -- (V2) -- (V3) -- (V4) -- (V1);
							\draw (V4) -- (V5) -- (V3);
							\draw[dashed] (V1) -- (V6) -- (V7) -- (V2);
						\end{tikzpicture}
						\caption{}
						\label{subfig:gamma:4:delta:3:3:3}
					\end{subfigure}
					\caption{Illustrations for the proof of Lemma~\ref{lem:gamma:4:delta:3:3:circuits}}
					\label{fig:gamma:4:delta:3:3}
				\end{figure}
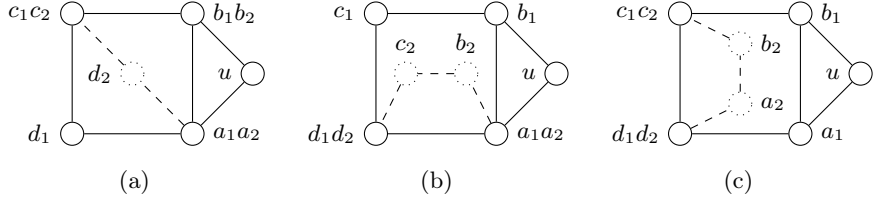
				
				\item Or $s_1 \cap s_2 = \{a_1 = a_2, c_1 = c_2, d_1 = d_2\}$ and $b_1 \neq b_2$. Similarly we can deduce that $b_2 = u$. We obtain the graph of Figure~\ref{subfig:gamma:4:delta:3:2:1}.
				\item The two other intersections are symmetrical cases.
			\end{itemize}
			\item If $s_1$ and $s_2$ have two common nodes as in Figure~\ref{subfig:gamma:4:delta:3:1:2}, then:
			\begin{itemize}
				\item $s_1 \cap s_2 = \{a_1 = a_2, b_1 = b_2\}$ and $c_1 \neq c_2, d_1 \neq d_2$. Then $c_2 = u$ and $d_2 = u$, otherwise the degree of $a_1$ or $b_1$ is 4. But $c_2$ cannot equal $d_2$.
				\item Or $s_1 \cap s_2 = \{a_1 = a_2, d_1 = d_2\}$ and $b_1 \neq b_2, c_1 \neq c_2$ (see Figure~\ref{subfig:gamma:4:delta:3:3:2}). Then $b_2 = u$ otherwise the degree of $a_1$ is 4, and $c_2 \neq u$.
				
				The edge $(b_2 = u, c_2)$ necessarily belongs either to a triangle or to another square of $C$, otherwise, we cannot have $\bigoplus T + \bigoplus C = 0$. Assuming it belongs to a triangle $t'$. As $u$ has 3 neighbors, $a_1, b_1$ and $c_2$, the third node of $t'$ is necessarily $a_1$ or $b_1$. This implies that $a_1$ or $b_1$ are linked to $c_2$. This is a contradiction as that nodes has then 4 neighbors.
				
				Consequently, the edge $(b_2 = u, c_2)$ belongs to a third square $s_3 = (a_3, b_3 = b_2 = u, c_3 = c_2, d_3)$. Again, as $u$ has already 3 neighbors, we have that $a_3 \in \{a_1 = a_2, b_1\}$. If $a_3 = a_1 = a_2$ then $d_3$ must be an existing neighbor of $a_1$, that is $b_1$ or $d_1 = d_2$. It cannot be $d_2$ otherwise $s_2 = s_3$. Then $b_1 = d_2$, $b_1$ has four neighbors and there is a contradiction. If, on the other hand, $a_3 = b_1$ then $d_3$ must be an existing neighbor of $b_1$, that is $a_1$ or $c_1$. In the first case, $a_1$ has four neighbors. In the second case, we obtain the graph of Figure~\ref{subfig:gamma:4:delta:3:2:2}.
				
				\item The case $s_1 \cap s_2 = \{b_1, c_1\}$ is symmetrical. 
				\item Or $s_1 \cap s_2 = \{c_1 = c_2, d_1 = d_2\}$ and $a_1 \neq a_2, b_1 \neq b_2$ (see Figure~\ref{subfig:gamma:4:delta:3:3:3}). The node $a_2$ cannot equal $a_1$ or $b_1$. It is then either $u$ or an additional node. The same property occurs for $b_2$. If $a_2 = u$ then $u$ has four neighbors $a_1, b_1$, $d_1 = d_2$ and $b_2$. As $b_2$ cannot equal any of the three first nodes, there is a contradiction. Similarly, there is a contradiction if $b_2 = u$. 
				
				The edge $(a_1, d_1)$ necessarily belongs to either a triangle or another square of $C$, otherwise, we cannot have $\bigoplus T + \bigoplus C = 0$. It cannot belong to a triangle as the two nodes have three neighbors and no common neighbor. There is then a third square $s_3 = (a_3 = a_1, b_3, c_3, d_3 = d_1)$ containing that edge. Then $b_3$ is a neighbor of $a_1$ and $c_3$ is a neighbor of $d_3$. As a consequence $b_3 = u$ and $c_3 = a_2$. Thus $u$ and $a_2$ are neighbors.
				
				Similarly, the edge $(b_1, c_1)$ belongs to a fourth square and $u$ and $b_2$ are neighbors, and then $u$ has four neighbors. This is a contradiction.
			\end{itemize}
		\end{itemize}
		
		Note finally that no square can be added to extend the circuits of Figure~\ref{fig:gamma:4:delta:3:2} as, by definition, $C$ is a minimal dependent set of cycles. 		
	\end{proof}
	
	\begin{theorem}
		\label{theo:gamma:4:delta:3:poly}
		\MCBI and \MMCBI are polynomial when $\gamma = 4$ and $\Delta = 3$.
	\end{theorem}
	\begin{proof}
		We use the following algorithm. First, we compute the squares of the list $L$ with Algorithm~\ref{algo:candidatecycle}. We then remove from $L$ any square that is generated by $T(G_i)$ for some $i \in \llbracket 1; k \rrbracket$. Then, we initialize an empty solution $B$ and, for each cycle $c \in L$, add $c$ to $B$ if $c \not\in \bigcup_{i = 1}^k span(B \cup T(G_i))$. We finally return $B$. 
		
		Let $B^*$ be an optimal solution and $B$ be the solution resulting from the algorithm. Let $s_{(j)}$ be the $j$-th added cycle of $L$ and let also $B_{(j)}$ be the set $B$ at the beginning of the j-th iteration of the algorithm (before adding $s_{(j)}$). Let $\alpha$ be the first index where $s_{(\alpha)} \in B^* \backslash B$ or $s_{(\alpha)} \in B \backslash B^*$. We assume that, among all the optimal solutions, $B^*$ is the solution maximizing the index $\alpha$.  
		
		If $s_{(\alpha)} \in B^* \backslash B$, then at the $\alpha$-th iteration of the algorithm, $s_{(\alpha)}$ is not added to $B$, meaning that $s_{(\alpha)} \in span(B \cup T(G_i))$ for some $i$, that is $\{s_{(i)}\} \cup B_{(i)}$ is not independent in $\mathcal{M}(G_i, 4)$. Note that, by definition of $\alpha$,  $B_{(\alpha)} \subseteq B^*$. And because $s_{(\alpha)} \in B^*$ then $\{s_{(\alpha)}\} \cup B_{(\alpha)} \subseteq B^*$, which is a contradiction.

		If, on the other hand, $s_{(\alpha)} \in B \backslash B^*$, then $B^* \cup \{s_{(\alpha)}\}$ contains a circuit $C$ in $\mathcal{M}(G_{i_C}, 4)$, with $s_{(\alpha)} \in C$, for some $i_C \in \llbracket 1 ; k \rrbracket$. We first make the assumption that $C$ is dependent in all the graphs. Let $s \in C \cap B^*$.
		\begin{itemize}
			\item Either the set $B(s) = B^* \backslash \{s\} \cup \{s_{(\alpha)}\}$ is feasible.
			\item Or for some $i \in \llbracket 1 ; k \rrbracket$, $B^* \backslash \{s\} \cup \{s_{(\alpha)}\}$ is dependent in $\mathcal{M}(G_i, 4)$. In that last case, there exists a circuit $C'$ of $\mathcal{M}(G_i, 4)$ with $C' \subset B(s)$ and $s_{(\alpha)} \in C'$. As $s \in C$ and $s \not\in C'$ then $C \neq C'$ ; and thus $(C \cup C') \backslash (C \cap C')$ is dependent in $\mathcal{M}(G_i, 4)$. However $(C \cup C') \backslash (C \cap C') \subset B^*$ and this is a contradiction. 
		\end{itemize}
		
		Therefore $B(s)$ is feasible for every $s \in C \cap B^*$. Note that $B(s)$ is also optimal. Let $j_1 < j_2 < \cdots < j_{|C|}$ be the $|C|$ indices of the squares of $C$ in $L$ (one of them is $\alpha$). By definition of $\alpha$, for every $j \leq \alpha$, $s_{(j)} \in B$. As $B$ is independent, $C  \not\subseteq B$: there exists a square $s_{(j)}$, with $j \in \llbracket \alpha + 1; j_{|C|} \rrbracket$ such that $s_{(j)} \not\in B$. The set $B(s_{(j)})$ is optimal and then contradicts the maximization of $\alpha$ by $B^*$.

		Consequently, $C$ is not dependent in all the graphs. Let $T$ be the set of triangles of $G_{i_C}$ such that $\bigoplus C = \bigoplus T$. There exists a graph $G_i$ such that $\bigcup T  \not\subseteq G_i$. Indeed, if we assume the contrary, we get that $C$ is dependent in all graphs. As a consequence, first $T \neq \emptyset$ and at least one edge in $\bigcup T$ is not part of any square in $C$. This implies that $C$ is not the circuit depicted by Figure~\ref{subfig:gamma:4:delta:3:2:2}.
		
		By Lemma~\ref{lem:gamma:4:delta:3:3:circuits}, $C$ is the circuit depicted by Figure~\ref{subfig:gamma:4:delta:3:2:1}. It contains $s_{(\alpha)}$ and another square $s_\epsilon$. Let $s_{(\alpha)} = (a_1, b_1, c_1, d_1)$ and $s_\epsilon = (a_2 = a_1, b_2 \neq b_1, c_2 = c_1, d_2 = d_1)$. We now prove that, in each graph $G \in \{G_1, G_2, \dots, G_k\}$, if $C$ is not a circuit in $\mathcal{M}(G, 4)$ then = $B^* \cup \{s_{(\alpha)}\}$ is independent in that matroid. Indeed, we have that $(b_1, b_2) \not\in G$. Also, if $B^* \cup \{s_{(\alpha)}\}$ is not independent, then there is another circuit $C' \subseteq B^* \cup \{s_{(\alpha)}\}$ in $\mathcal{M}(G, 4)$ containing $s_{(\alpha)}$. This circuit $C'$ is again depicted by Figure~\ref{subfig:gamma:4:delta:3:2:1}, similar to $C$. Let $s_{(\lambda)} = (a_3, b_3, c_3, d_3)$ be the second square of $C'$. The four possible cases are considered.

		\begin{itemize}
			\item $b_1 = b_3, c_1 = c_3, d_1 = d_3$ and $(a_1, a_3) \in G$. Note that $a_3 \neq b_2$ as the edge $(b_1 = b_3, b_2)$ is not in $G$. However in that case $a_1$ has four neighbors.
			\item The case $a_1 = a_3, b_1 = b_3, d_1 = d_3$ and $(c_1, c_3) \in G$ is symmetrical.
			\item $a_1 = a_3, c_1 = c_3, d_1 = d_3$ and $(b_1, b_3) \in G$. Note that $b_3 \neq b_2$ as $(b_1, b_2) \not\in G$. However in that case $a_1$ has four neighbors.
			\item $a_1 = a_3, b_1 = b_3, c_1 = c_3$ and $(d_1, d_3) \in G$. If $d_3 = b_2$, then $(d_3 = b_2, d_1 = d_2)$ is a diagonal of $s_{(\varepsilon)}$. However, if $d_3 \neq b_2$ then $a_1$ has four neighbors.
		\end{itemize}
		The existence of $C'$ is a contradiction, $B^* \cup \{s_{(\alpha)}\}$ is independent in $\mathcal{M}(G, 4)$.
		
		As a consequence, $B^* \backslash \{s_{(\varepsilon)}\} \cup \{s_{(\alpha)}\}$ is feasible and optimal. However, as $B$ is independent and contains $s_{(\alpha)}$, then $s_{(\varepsilon)} \not\in B$, this means that $\varepsilon > \alpha$. The optimality of $B^* \backslash \{s_{(\varepsilon)}\} \cup \{s_{(\alpha)}\}$ contradicts the fact that $B^*$ maximises the value of $\alpha$. As a conclusion no such square $s_{(\alpha)}$ exists and $B$ is optimal. 		
	\end{proof}
	
	
	\section{Concluding remarks}
	\label{sec:conclu}
	
	This paper introduces the problem of maximizing the intersection of minimum cycle bases of graphs and studies the complexity based on four natural parameters. Several questions remain open regarding the minimum value of $k$ for which inapproximability and W[1]-hardness hold. Furthermore, from a chemical perspective, one could argue that the difference between the sets of edges of the graphs (the edit distance) may be small relative to $k$, $\Delta$, and $\gamma$. This observation may lead to the discovery of new tractable algorithms for the problem.
	
	\bibliographystyle{splncs04}
	\bibliography{biblio}

	\appendix
	
	\section{Proofs of Lemmas~\ref{lem:exchangecyclesinbasis}, \ref{lem:generatecyclebylowerweight}, \ref{lem:basisexchangeaxiom} and \ref{lem:circuit}}
	
	\label{app:proofs}
	
	\exchangecyclesinbasis*
	\begin{proof}
		Let $B' = (B \backslash \{c_1\}) \cup c_2$. Let $d$ be any cycle of $G$. As $\lambda_B(c_1, c_2) = 1$, then there exists $D \subset B \backslash \{c_1\}$ such that $c_2 = c_1 \oplus \bigoplus D$. Thus, $c_1 = c_2 \oplus \bigoplus D$. As a consequence, $c_1$ is generated by $B'$. 
		
		Then any cycle that is generated by $B$ is also generated by $B'$. As $|B'| = |B|$, it is a cycle basis. 		
	\end{proof}
	
	We prove two intermediates lemma to prove Lemmas~\ref{lem:generatecyclebylowerweight} and \ref{lem:basisexchangeaxiom}.
	
	\begin{lemma}
		\label{lem:generatecyclebylowerweight:1}
		If $B \in \mcb{G}$ then for $c_1, c_2$ with $c_1 \in B$ and $c_2 \not\in B$ such that $\lambda_B(c_1, c_2) = 1$, we have $\omega(c_1) \leq \omega(c_2)$. 
	\end{lemma}
	\begin{proof}
		Let $B \in \mcb{G}$. Then by definition $B$ is a cycle basis. Let $c_1 \in B$ and $c_2 \not\in B$ with $\lambda_{B}(c_1, c_2) = 1$. By Lemma~\ref{lem:exchangecyclesinbasis}, $(B \backslash \{c_1\}) \cup \{c_2\}$ is a cycle basis. Then, if $\omega(c_2) < \omega(c_1)$, then the weight of $B'$ is lower than the one of $B$, this contradict the optimality of $B$. 
	\end{proof}
	
	\begin{lemma}
		\label{lem:generatecyclebylowerweight:2}
		Let $B_1, B_2$ be two cycle bases of $G$ such that
		\begin{itemize}
			\item for each $B \in \{B_1, B_2\}$, for every $c_1, c_2$ where $c_1 \in B$ and $c_2 \not\in B$ such that $\lambda_B(c_1, c_2) = 1$, we have $\omega(c_1) \leq \omega(c_2)$
		\end{itemize}
		then 
		\begin{itemize}
			\item for every $c_2 \in B_2 \backslash B_1$, there exists $c_1 \in B_1 \backslash B_2$ such that $(B_2 \backslash \{c_2\}) \cup \{c_1\}$ is a cycle basis with same weight as $B_2$.
		\end{itemize}
	\end{lemma}
	\begin{proof}
		The result may be proved with Lemma~\ref{lem:exchangecyclesinbasis} if we demonstrate there exists $c_1 \in B_1$ with $\omega(c_1) = \omega(c_2)$ and $\lambda_{B_2}(c_2, c_1) = 1$. 
		
		Let $D_1 = \{d \in B_1 \cap B_2 | \lambda_{B_1}(d, c_2) = 1\}$, $D_2 = \{d \in B_1 \backslash B_2 | \lambda_{B_1}(d, c_2) = 1 \text{ and } \omega(d) < \omega(c_2)\}$ and $D_3 = \{d \in B_1 \backslash B_2 | \lambda_{B_1}(d, c_2) = 1 \text{ and } \omega(d) = \omega(c_2)\}$. By the hypothesis on $B_1$, no cycle of $B_1$ with weight greater than $c_2$ generates $c_2$. Then $c_2 = \bigoplus D_1 \oplus \bigoplus D_2 \oplus \bigoplus D_3$. 
		
		By the hypothesis on $B_2$, for all $d \in D_2$ and $e \in B_2$ such that $\lambda_{B_2}(e, d) = 1$, $\omega(e) \leq \omega(d) < \omega(c_2)$ thus $e \neq c_2$: for all $d \in D_2$, $\lambda_{B_2}(c_2, d) = 0$. Note also that $c_2 \not\in D_1$ as $D_1 \subseteq B_1$. Thus, if we have for all $d \in D_3$, $\lambda_{B_2}(c_2, d) = 0$ then $$c_2 = \bigoplus D_1 \oplus \bigoplus\limits_{d \in D_2 \cup D_3} \bigoplus\limits_{\substack{e \in B_2 \backslash c_2\\\lambda_{B_2}(e, d) = 1}} e$$
		
		As $c_2$ does not belong to the right side of the equation, $B_2$ is not linearly independent. This is a contradiction. There is then $d \in D_3$ such that $\lambda_{B_2}(c_2, d)~=~1$. By Lemma~\ref{lem:exchangecyclesinbasis}, $(B_2 \backslash \{c_2\}) \cup \{c_1\}$ is a cycle basis with same weight as $B_2$. As $d \in B_1 \backslash B_2$, the lemma is proved.
	\end{proof}
	
	\generatecyclebylowerweight*
	\begin{proof}
		The forward direction is proved by Lemma~\ref{lem:generatecyclebylowerweight:1}. 
		
		Now given a basis $B_1$ satisfying the second property and $B_2$ be a minimum cycle basis. By Lemma~\ref{lem:generatecyclebylowerweight:1}, then we can apply Lemma~\ref{lem:generatecyclebylowerweight:2} with $B_1$ and $B_2$. Given $c_2 \in B_2 \backslash B_1$, there exists $c_1 \in B_1 \backslash B_2$ such that $(B_2 \backslash \{c_2\}) \cup \{c_1\}  \in \mcb{G}$. We can then successively exchange all the cycles of $B_2$ with cycle of $B_1$ until the two coincide. As the new basis is still a minimum cycle basis, this demonstrates that $B_1$ is a minimum cycle basis. 
	\end{proof}
	
	\basisexchangeaxiom*
	\begin{proof}
		This is a direct consequence of Lemmas~\ref{lem:generatecyclebylowerweight:1} and \ref{lem:generatecyclebylowerweight:2}.
	\end{proof}
	
	\circuit*
	\begin{proof}
		We denote by $S_{l-1}$ the set $span(\text{cycle } c | \omega(c) \leq l - 1)$.
		
		We first prove the necessary condition. Let $B$ be a subset of cycles of size $l$ such that $B$ is independent in $\mathcal{M}(G, l)$. Then $B$ is independent in $\mathcal{M}(G)$: there exists a minimum cycle basis $B'$ in $G$ containing $B$.
		
		We assume there exists $D \subseteq B$ such that $\bigoplus D \in S_{l-1}$. Then there exists $E$ containing cycles of $G$ with weight at most $l - 1$ such that $\bigoplus D = \bigoplus E$. By Lemma~\ref{lem:generatecyclebylowerweight}, for all $e \in E$ and $f \in B'$ such that $\lambda_{B'}(f, e) = 1$, we have $\omega(f) \leq \omega(e) = l - 1$. This means that $f \not\in D$. Consequently $$\bigoplus D = \bigoplus\limits_{e \in E} \bigoplus\limits_{\substack{f \in B' \backslash D \\ \lambda_{B'}(f, e) = 1}} f$$
		
		As no cycle of $D$ belongs to the right side of the equation, $B'$ is linearly dependent, which is a contradiction. This proves that no such set $D$ exists.
		
		Now, we prove the sufficient condition, assuming that for all $D \subset B$, $\bigoplus D \not\in S_{l-1}$. Then first $B$ is linearly independent (otherwise for some $D$, we would have $\bigoplus D = 0 \in S_{l-1}$). Let $B^* \in \mcb{G}$. We now prove that if $B  \not\subseteq B^*$ then for all $d \in B \backslash B^*$, there exists $c \in B^*$ such that $B^* \backslash \{c\} \cup \{d\} \in \mcb{G}$. 
		
		Let $d \in B \backslash B^*$. If, for all $c \in B^*$ such that $\lambda_{B^*}(c, d) = 1$, we have $\omega(c) \neq l$, then by Lemma~\ref{lem:generatecyclebylowerweight}, $\omega(c) \leq l - 1$. Consequently, $d \in S_{l-1}$, this contradicts the hypothesis on $B$. Given $c$ with $\lambda_{B^*}(c, d) = 1$ and $\omega(c) = l$, Lemma~\ref{lem:exchangecyclesinbasis} proves $B^* \backslash \{c\} \cup \{d\} \in \mcb{G}$. By doing such exchanges we eventually get a minimum cycle basis containing $B$.

	\end{proof}
	
\end{document}